\documentclass{jocg}
\usepackage{amsmath,amssymb,amsthm}
\usepackage{graphicx, hyperref}
\usepackage{cite}
\usepackage{aliascnt}
\usepackage{microtype}

\theoremstyle{plain}
\newtheorem{theorem}{Theorem}

\newaliascnt{lemma}{theorem}
\newtheorem{lemma}[lemma]{Lemma}
\aliascntresetthe{lemma}

\newaliascnt{corollary}{theorem}

\aliascntresetthe{corollary}

\theoremstyle{definition}
\newaliascnt{definition}{theorem}

\aliascntresetthe{definition}

\usepackage{makecell} 

\newcommand\NP{$\mathsf{NP}$}
\newcommand\PSPACE{$\mathsf{PSPACE}$}

\let\epsilon=\varepsilon

\title{\MakeUppercase{Flat Foldings of Plane Graphs\penalty-1000 with Prescribed Angles and Edge Lengths}}
\author{%
Zachary Abel,\thanks{\affil{Department of Mathematics, MIT},
\email{zabel@mit.edu}}\,
Erik D. Demaine,\thanks{\affil{MIT Computer Science and Artificial Intelligence Lab},
\email{edemaine@mit.edu}}\,
Martin L. Demaine,\thanks{\affil{MIT Computer Science and Artificial Intelligence Lab},
\email{mdemaine@mit.edu}}\,
David Eppstein,\thanks{\affil{Department of Computer Science, University of California, Irvine},
\email{eppstein@uci.edu}}\,
Anna Lubiw,\thanks{\affil{David R. Cheriton School of Computer Science, University of Waterloo},
\email{alubiw@uwaterloo.ca}}\penalty-1000
and Ryuhei Uehara\thanks{\affil{School of Information Science, Japan Advanced Institute of Science and Technology},
\email{uehara@jaist.ac.jp}}%
}

\begin{document}
\maketitle

\begin{abstract}
When can a plane graph with
prescribed edge lengths and prescribed angles (from among $\{0,180^\circ,
360^\circ$\}) be folded flat to lie in an infinitesimally thin line,
without crossings?
This problem generalizes the classic theory of single-vertex flat origami
with prescribed mountain-valley assignment,
which corresponds to the case of a cycle graph.
We characterize such flat-foldable plane graphs by two obviously necessary
but also sufficient conditions, proving a conjecture made in 2001:
the angles at each vertex should sum to $360^\circ$,
and every face of the graph must itself be flat foldable.
This characterization leads to a linear-time algorithm for testing
flat foldability of plane graphs with prescribed edge lengths and angles, and a polynomial-time algorithm for counting the number of distinct folded states.
\end{abstract}

\pagestyle{plain}

\section{Introduction}

The modern field of origami mathematics began in the late 1980s with the goal
of characterizing flat-foldable crease patterns, i.e., which plane graphs form
the crease lines in a flat folding of a piece of paper~\cite{DemORo-GFA-07}.
This problem turns out to be \NP-complete in the general case,
with or without an assignment of which folds are mountains and which are
valleys~\cite{BerHay-SODA-96}.

\begin{figure}[t]
\centering
\includegraphics[width=5in]{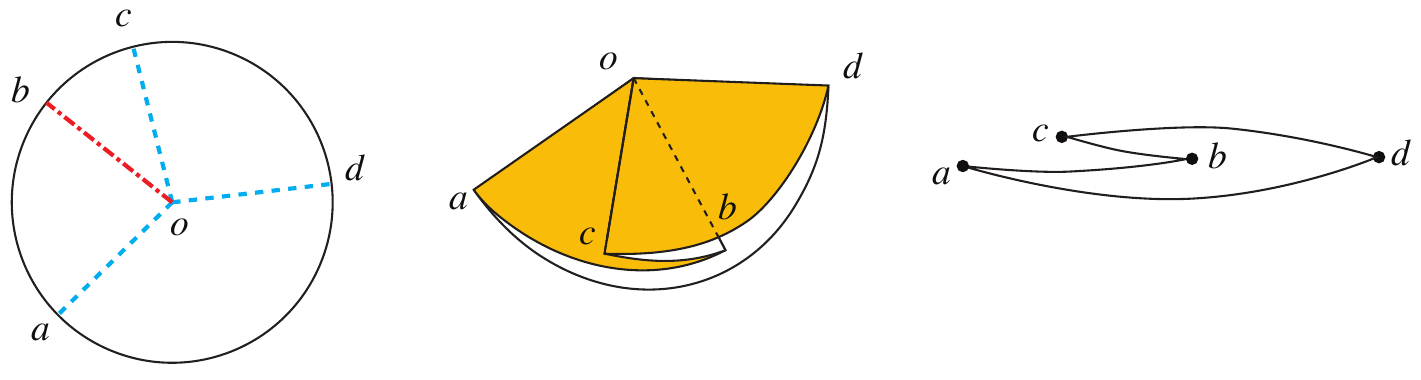}
\caption{Flat folding at a single vertex on a disc reduces to the problem of folding a polygon onto a line.}
\label{fig:one-cycle}
\end{figure}

On the other hand, flat foldability can be solved in polynomial time for crease
patterns with just a single vertex (thus characterizing the
local behavior of a vertex in a larger graph).
By slicing the paper with a small sphere centered at the single vertex
(the geometric \emph{link} of the vertex),
single-vertex flat foldability reduces to the 1D problem
of folding a polygon (closed polygonal chain) onto a line;
see \autoref{fig:one-cycle}.
This problem can be solved by a greedy algorithm that repeatedly folds both
ends of a shortest edge with opposite fold directions (mountain and
valley)---either because such directions have already been pre-assigned or,
if the mountain-valley assignment is not given, by making such an assignment
\cite{BerHay-SODA-96,Hul-O3-02,ArkBenDem-CGTA-04,DemORo-GFA-07}.
The spherical, self-touching Carpenter's Rule Theorem
\cite{ConDemRot-DCG-03,StrWhi-JCDCG-04,AbbDemGas-arXiv-09}
implies that any flat-folded single-vertex origami can be reached from the
unfolded piece of paper by a continuous motion that avoids bending or folding
the uncreased parts of the paper.

In practical applications of folding beyond origami, the object being folded
may not be a single flat sheet, but rather some 2D polyhedral cell complex
with nonmanifold topology (for instance, more than two facets joined at an edge, or a vertex with a disconnected link).
Flat foldability of such complexes is no easier than the origami case,
but again we can hope for reduced complexity when a complex has only a single vertex that should also appear as a convex vertex of the flat-folded state.
As with one-vertex origami, we can reduce the problem to 1D
by slicing with a small sphere centered at the vertex---now resulting in a general plane graph rather than a simple cycle---and
asking whether this graph can be flattened onto a
line~\cite{AbeDemDem-IJCGA-13}; see \autoref{fig:two-cycles}.
In this way, the problem of flat-folding single-vertex complexes can be
reduced to finding embeddings of a given plane graph onto a line.
The same slicing technique works more generally for flat-foldings of single-vertex complexes that extend for less than a full circle around the vertex, without requiring the vertex to be convex in the flat-folded state.

\begin{figure}[t]
\centering
\includegraphics[width=6in]{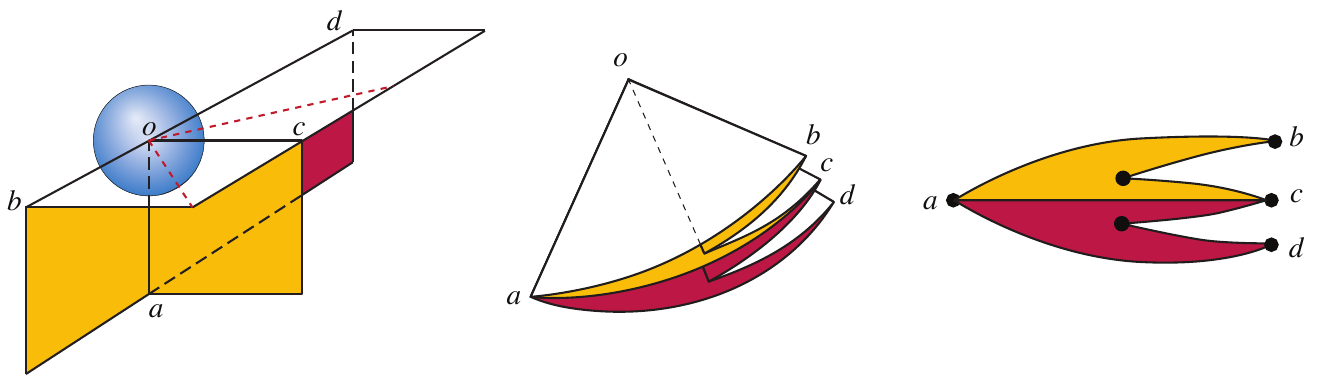}
\caption{Flat folding a two-dimensional cell complex with a single vertex reduces to the problem of folding a plane graph onto a line.}
\label{fig:two-cycles}
\end{figure}

It is this problem that we study here:
given a plane graph with specified edge lengths,
does it have a straight-line plane embedding with all vertices arbitrarily
close to a given line and with all edges arbitrarily close to their specified
lengths?
In the version of the problem we study,
we are additionally given a specification of whether the angle between
every two consecutive edges at each vertex is a \emph{mountain fold}
(the angle is arbitrarily close to $360^\circ$), a \emph{valley fold}
(the angle is arbitrarily close to $0$), or \emph{flat}
(the angle is arbitrarily close to $180^\circ$).
Without this information, the problem of testing whether a given plane graph
can be folded flat with specified edge lengths (allowing angles of $180^\circ$)
is weakly \NP-complete, even for graphs that are just simple cycles,
by a straightforward reduction from the subset sum problem.
For general plane graphs, the problem becomes strongly
\NP-complete~\cite{AbeDemDem-IJCGA-13}.
Therefore, we concentrate in this paper on the version of the problem with
given angle assignments, posed as an open problem in~\cite{AbeDemDem-IJCGA-13}.

\subsection{New Results}

We show that it is possible to test in linear time whether a given plane graph,
with given edge lengths and angle assignment, can be folded flat;
refer to \autoref{tbl:variations}.
Additionally, in polynomial time, we can count the number of
combinatorially distinct flat foldings.
Intuitively, two foldings are combinatorially equivalent when the same pairs of edges touch each other at the same positions in the folding, and combinatorially distinct otherwise.
 for a more careful definition of combinatorial distinctness, see \autoref{sec:defs}.

\begin{table}[t]
\centering
\tabcolsep=0.7em
\def\arraystretch{1.3}
\begin{tabular}{r|cc}
&~~Flat angles forbidden~~&~~Flat angles allowed~~\\
\hline
~~Angle assignment given~~&Linear time (new)&Linear time (new)\\
~~Angle assignment unspecified~~&Open&\NP-complete~\cite{AbeDemDem-IJCGA-13}\\
\end{tabular}
\medskip
\caption{Complexity of flat folding a plane graph, by input model}
\label{tbl:variations}
\end{table}

Our algorithms are based on a new characterization of flat-foldable graphs: a flat folding exists if and only if the angles at each vertex sum to $360^\circ$ and each individual face in the given graph can be folded flat. 
Even stronger, we show that
independent flat foldings of the interior of each face can always be combined into a flat folding of the whole graph. Here, again, foldings of the interior of a face are considered equivalent when the same pairs of interior sides of edges touch each other in the same way, and distinct otherwise, ignoring the contacts of exterior sides of edges of the face or between edges that are not part of the face.
\autoref{fig:in-vs-out} shows an example of this combination of face foldings.

 A form of the theorem was conjectured in 2001 by Ilya Baran, Erik Demaine, and Martin Demaine, but not proved until now; it contradicts the intuitive (but false) idea that, for faces with ambiguous spiraling shapes, each face must be folded consistently with its neighboring faces. With this theorem in hand, our algorithms for constructing and counting folded states follow by using a greedy ``crimping'' strategy for flat-folding simple cycles~\cite{BerHay-SODA-96,ArkBenDem-CGTA-04,DemORo-GFA-07} and by using dynamic programming to count cycles within each face.

\begin{figure}[t]
\centering
\includegraphics[width=0.75\textwidth]{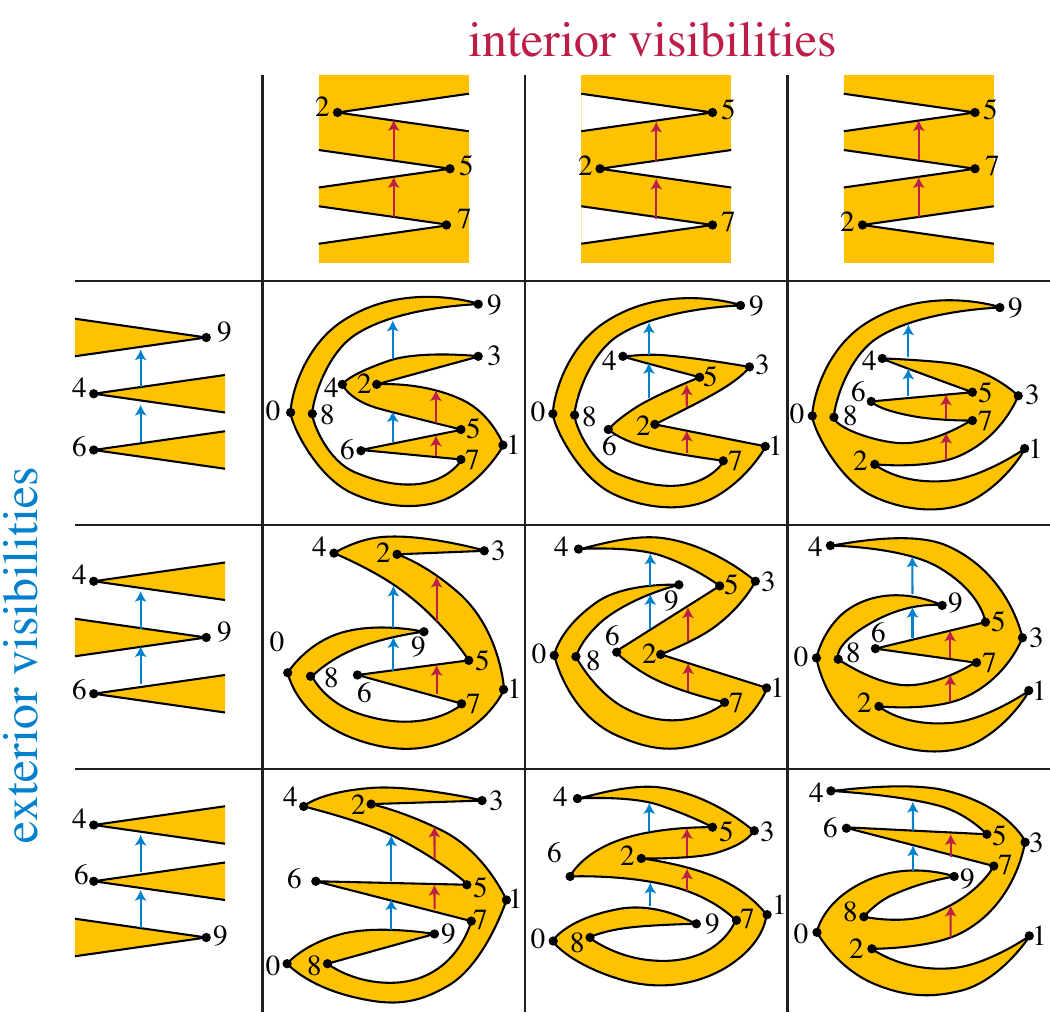}
\caption{A planar graph with two faces, each of which can be flat-folded to give three different patterns of vertical visibility (or ``touching'') within it. These patterns can be combined independently, giving nine flat-foldings of the whole graph. The nine flat-foldings are depicted with the correct above-below relation of edges, but (for legibility) not the correct left-right relation of vertices.}
\label{fig:in-vs-out}
\end{figure}

Our characterization necessarily concerns flat folded states in isolation from each other.
We do not consider the ability to reach a given state by continuous
folding motions from a given (nonflat) configuration.
As shown by past work, even for trees,
there exist locked states that cannot be continuously moved to a 
flat folded state~\cite{BieDemDem-DAM-02,BalChaDemDemIacLiuPoo-WADS-09},
and testing the existence of a continuous motion between two states is
\PSPACE-complete~\cite{AltKnaRot-TTGG-04}.

We leave open the problem of finding a flat folded state for a graph in which
the planar embedding and edge lengths are preassigned, and angles of
$180^\circ$ are forbidden, but the choice of which angles at each vertex are
$0$ and which are $360^\circ$ is left free
(bottom-left cell of \autoref{tbl:variations}).
Even for trees, this open problem appears to be nontrivial; see
\autoref{fig:impossible-tree} and~\cite{EstFowKob-CGTA-09}.

\begin{figure}[t]
\centering
\includegraphics[width=5in]{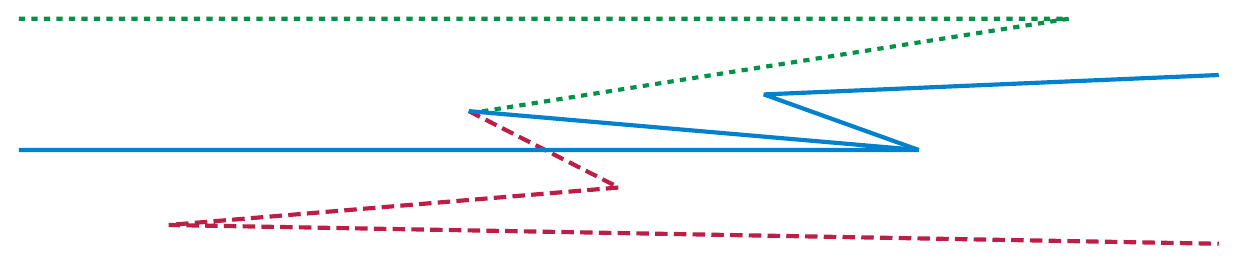}
\caption{A tree with fixed edge lengths (the Euclidean lengths of the line segments shown) that (when angles equal to $180^\circ$ are forbidden) has no 
flat folding, regardless of planar embedding or angle assignment}
\label{fig:impossible-tree}
\end{figure}

\subsection{Related Work}
\label{sec:related}

There has been intensive study of
straight-line drawings of graphs with specified edge lengths and/or specified angles between consecutive edges in a cyclic ordering of edges around each vertex.
If only edge lengths are specified then---whether the drawing must be planar or not---the problem is \NP-hard~\cite{Cabello-Demaine-Rote,Saxe}, or worse~\cite{Schaefer}.
It is also \NP-hard to draw a plane graph with specified angles~\cite{Garg}.
If both edge lengths and angles are specified then 
the drawing is uniquely determined and easy to construct, except in situations like ours 
where coincident edges give rise to ambiguities.

There are a number of results for special cases that have a similar flavor to ours, in that  the whole plane graph can be realized if and only if each face can.  
We now describe some of these special cases, most of which
arise as the prelude to finding an appropriate angle assignment.

\medskip\noindent
{\bf Upward Planarity.}
A directed acyclic graph (DAG) is \emph{upward planar}~\cite{Garg1995upward} if it has a planar drawing in which each edge is drawn as an increasing $y$-monotone curve. 
Recognizing upward planar graphs is \NP-hard~\cite{GarTam-SJC-01} but Bertolazzi et al.~\cite{Ber-Algo-94} gave a polynomial time algorithm for the special case of a plane graph whose cyclic order of edges around each vertex is prescribed. 
The main issue in their solution is to distinguish ``small'' versus ``large'' angles;
if an upward planar drawing is flattened onto a vertical line, then its small and large angles correspond to our valley and mountain folds.
The angle assignment is forced except at vertices with only incoming or only outgoing edges, where exactly one angle should be large and the rest small.
Bertolazzi et al. used network flows to determine these angles. To prove their algorithm's correctness, they showed that a graph with a given angle assignment has an upward planar drawing if and only if each face cycle has an upward planar drawing.  The condition for drawing a single face, given an angle assignment, is particularly simple: 
an acyclic orientation of a cycle has an upward planar drawing if and only if it has two more small than large angles.
Their proof also shows that embedding choices for the faces can be combined arbitrarily. 

\medskip\noindent
{\bf Level Planarity.}
Our flat folding problem differs from upward planarity in that we have assigned edge lengths as well as assigned angles.
This makes it more similar to the problem of \emph{level planarity}~\cite{DibNar-TSMC-88,JunLeiMut-GD-98,HarHea-GD-07}.
The input to this problem is a \emph{leveled} directed acyclic graph: a DAG whose vertices have been partitioned into a sequence of levels (independent sets of its vertices), with all edges directed from earlier to later levels. The goal is to find an upward planar embedding that places the vertices of each level on a horizontal line~\cite{PacTot-JGT-04}. 
This problem has a linear time solution~\cite{JunLeiMut-GD-98} based on PQ-trees.  When the cyclic order of edges around vertices is specified (and in fact for more general constraints) there is a quadratic time solution based on solving systems of binary equations~\cite{HarHea-GD-07}.  


The input to our folding problem may be interpreted as a leveled plane DAG.  (Since our convention is to flatten to a horizontal line, we will map to a level planarity problem with levels progressing rightward rather than upward---this is a superficial difference.)  
Arbitrarily choose an $x$-coordinate for one vertex in each connected component of the graph and a direction (left-to-right) for one edge incident to that vertex.  These choices can be propagated to all the vertices and edges using the specified edge lengths and angles.  The set of vertices at a given $x$-coordinate constitute a level, giving us an input to the level planarity problem with a prescribed plane embedding. 
However, the embeddings we seek in the folding problem are not the same as level planar embeddings. In a level planar embedding, vertices within a single level must be linearly ordered by the second coordinate value. In contrast, in the folding problem a vertex that has only incoming or only outgoing edges may be nested between two adjacent edges of another vertex at the same level. A four-vertex cycle, oriented with alternating edge directions, illustrates the difference between these two types of embedding: it is not level planar, but it still has a flat folding with three mountain folds and one valley fold, corresponding to the usual way of folding a square sheet of paper into quarters (\autoref{fig:quarters}).

\begin{figure}[t]
\centering
\includegraphics[width=4.5in]{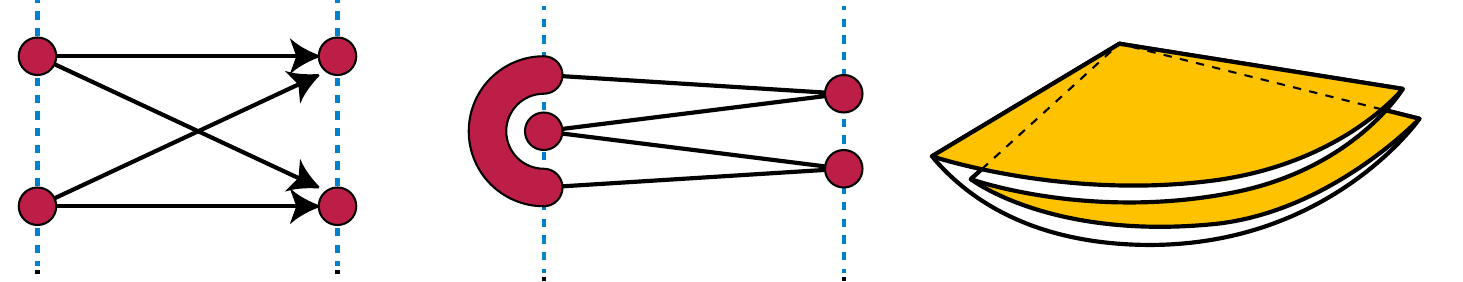}
\caption{A four-vertex cycle with vertices alternating between two levels is not level planar, but can be folded flat, representing a sheet of paper folded into quarters.}
\label{fig:quarters}
\end{figure}

We have therefore been unable to apply level planarity algorithms to solve our flattening problem.
On the other hand, our algorithm can be used to test level planarity of a plane graph (with a linear order of incoming and outgoing edges at each vertex) in linear time, faster than the practical quadratic algorithm of Harrigan and Healy~\cite{HarHea-GD-07} and matching the known linear time algorithm of J\"unger, Leipert, and Mutzel~\cite{JunLeiMut-GD-98}.  Given an input to the level planarity problem, we assign increasing coordinates to the levels, and assign the length of an edge to be the difference in coordinates between the levels of its endpoints.  Mountain/valley/flat angles are determined from the level assignment.  Finally, to preclude the nesting of vertices that is allowed in flattening but not in level planarity, we add an extra edge incident to each vertex that has only incoming or only outgoing edges:  if vertex $v$ on level $i$ has only outgoing edges, we add a new incoming edge from a new level just before $i$.   The resulting plane graph has a flat-folding respecting the angles and edge-lengths if and only if the original has a level planar drawing. 
From this we also obtain the result that a leveled plane graph has a level planar drawing if and only if each cycle does.

\medskip\noindent
{\bf Rectilinear Planarity.}
In our flat folding problem, the angles are multiples of $180^\circ$.  
When angles are multiples of $90^\circ$, we arrive at the important problems of orthogonal and rectilinear graph drawing~\cite{Duncan}.
A graph is \emph{rectilinear planar} if it can be drawn in the plane so that every edge is a horizontal or vertical line segment.  Coincident edges are forbidden, so the graph must have maximum degree 4.
This problem is \NP-complete in general~\cite{GarTam-SJC-01} but---as with upward planarity---there is a polynomial time algorithm, due to Tamassia~\cite{Tam-SJC-87}, if the cyclic order of edges around vertices is prescribed.  
Again, the main issue is to find an assignment of angles or, equivalently, a labeling of the edges incident to a vertex with distinct labels from the set ${U,D,L,R}$, where $U$ stands for ``Up'', etc.
Tamassia finds the angles using network flows  (in fact, he solves a more general problem of minimizing the number of bends in the drawing). 
At the heart of this method is the result that, 
given an angle assignment that is locally consistent (i.e., the angles at every vertex sum to $360^\circ$), the graph has a rectilinear planar drawing if and only if  each face cycle does~\cite{Tam-SJC-87}.
As in the other cases we have discussed, the proof shows the stronger result that embedding choices for the faces can be combined arbitrarily.  A cycle with an angle assignment has a rectilinear planar drawing if and only if the number of right turns minus the number of left turns is 4 in a clockwise traversal inside the cycle~\cite{Tam-SJC-87,Vijayan1985}.

Our result on flat folding can be used to prove an extension of the above result to rectilinear graph drawings with angles specified and with 
lengths assigned to the ``horizontal'' edges.  (Note that the angle information allows us to distinguish the two classes of edges, although it is arbitrary which class is horizontal and which is vertical.)
Given a rectilinear plane graph, contract all vertical edges, assigning angles  of $0, 180^\circ, 360^\circ$ in the obvious way.  Finally, in order to avoid ``nested'' vertices at the same coordinate (as in \autoref{fig:quarters}), we use the same trick of adding an extra edge incident to each vertex that has only incoming or only outgoing edges.
We claim that the resulting multi-graph has a flat-folding if and only if the original has a rectilinear planar drawing with horizontal edges of the specified lengths.  For the non-trivial direction, suppose we have a flat folding of the constructed graph.  We must expand each vertex to a vertical line segment with the horizontal edges touching the line segment in a way that is consistent with the original graph.  This can easily be done in a left to right sweep.

From this reduction we obtain the following result. For the purposes of this theorem, a \emph{face cycle} means a cyclic sequence alternating between vertices and edges, allowing both repeated vertices and repeated edges, obtained by tracing the boundary of a face in an embedding. Later in \autoref{sec:euler} we will define an equivalent notion in which the faces are represented by simple cycles (not allowing repetitions).

\begin{theorem}
If $G$ is a plane graph of maximum degree 4 with
specified angles that are multiples of $90^\circ$ and sum to $360^\circ$ at each vertex, and $G$ has 
lengths assigned to its horizontal edges, then $G$ has a rectilinear planar drawing realizing these angles and edge lengths if and only if every face cycle has a rectilinear planar drawing realizing the angles and lengths.  
\end{theorem}

\section{Definitions}
\label{sec:defs}

We assume that the input is a plane graph (a graph already equipped with a planar embedding),
with a length assignment on each edge and an angle assignment at each vertex.
We also assume that these length and angle assignments are consistent: the angles around each vertex should add to $360^\circ$, and following two different paths from any one vertex to any other should give the same signed sum of lengths.
The embedding specifies the clockwise ordering of edges around each face cycle and around each vertex of the graph, and also specifies which face cycle is outermost in the embedding; the folding should respect this embedding.

\subsection{Self-touching configurations}

Following previous work in this area~\cite{ConDemRot-PK-02,Rib-PhD-06} we formalize the notion of a flat folding using \emph{self-touching configurations}. Intuitively, these are planar embeddings in which edges and vertices are allowed to be infinitesimally close to each other.
A one-dimensional self-touching configuration of a graph $G$ consists of a mapping from $G$ to a path graph $H$ that maps vertices of $G$ to vertices of $H$ and edges of $G$ to paths in $H$, together with a \emph{magnified view} of each vertex and edge of $H$ that describes the local connectivity of the image of $G$ at that point. (The previous papers~\cite{ConDemRot-PK-02,Rib-PhD-06} also defined two-dimensional self-touching configurations, using mappings to a graph $H$ that is not a path graph, but we will not use these.)

More explicitly, given any path graph $H$, form a linear layout of $H$ in which the vertices are replaced by disjoint disks along a horizontal line, and the edges are replaced by finite-width rectangular corridors between these disks, as shown by the blue disks and corridors in the upper right of \autoref{fig:selftouch}. Then an \emph{expanded drawing} for a planar graph $G$ on $H$ is a planar drawing of $G$ such that:
\begin{itemize}
\item Each vertex of $G$ is placed within one of the disks of $H$.
\item Each disk of $H$ contains at least one vertex of $G$.
\item Each edge of $G$ is represented by a curve that stays within the disks and corridors of $H$, 
traverses
at least one corridor, and is monotone within each of the corridors that it traverses.
\end{itemize}
By monotonicity, each edge can traverse
any given corridor at most once, and cannot turn back on itself within any corridor.
Because the edges are required to be monotone within each corridor, they may be linearly ordered by their crossings with any vertical line.

We define a one-dimensional self-touching configuration to be a combinatorial abstraction of an expanded drawing,
consisting of a mapping from vertices of $G$ to vertices of $H$ and edges of $G$ to paths of one or more edges of $H$,
together with the vertical ordering of the edges of $G$ along each corridor (an edge of $H$).
We call this vertical ordering the \emph{magnified view} of the edge of $H$.
Correspondingly, we call the part of the drawing of $G$ within each of the vertex disks of $H$
the magnified view of that vertex, although these will be less important for us than the magnified views of edges.

In a self-touching configuration, the \emph{multiplicity} of an edge in $H$ is a non-negative integer, the number of different edges of $G$ that map to it.
If $G$ is connected, every edge of $H$ will necessarily have positive multiplicity, but we also allow zero multiplicity when $G$ is disconnected.

\begin{figure}[t]
\centering\includegraphics[width=\textwidth]{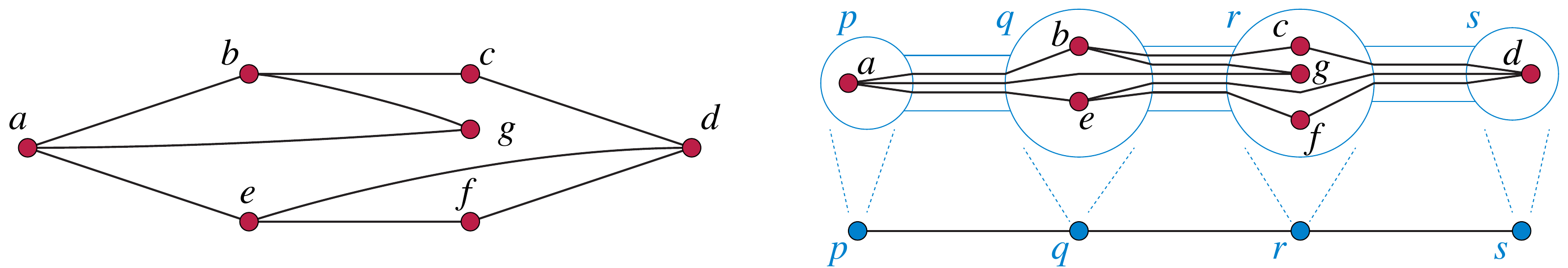}
\caption{A flat-folding of a seven-vertex graph $G$ (left), described as a self-touching configuration in which $G$ is mapped onto a four-vertex path $H$ (right), shown with magnified views of its edges and vertices of $H$.}
\label{fig:selftouch}
\end{figure}

\subsection{Flat foldings}

We may now define a \emph{flat folding} of a graph $G$ with assigned edge lengths
to be a self-touching configuration in which, in addition,
the edges of the path graph $H$ have assigned lengths,
and the length of each edge in $G$ equals the length of the corresponding path in $H$.
A \emph{face} of a flat folding or self-touching configuration is a cycle formed by a face of the expanded drawing. The \emph{angle} formed by a pair of incident edges in a flat folding is one of three values, $0$, $180^\circ$, or $360^\circ$, accordingly as the two edges extend in the same direction from their common endpoint with the face lying between them, the two edges extend in opposite directions, or the two edges extend in the same direction with the face on both sides. 
An \emph{angle assignment} to a plane graph is an assignment of the values $0$, $180^\circ$, or $360^\circ$ to each of its angles, regardless of whether this assignment is
compatible 
with a flat-folding of the graph.
An angle assignment is \emph{consistent} if the angles sum to $360^\circ$ at every vertex.

We define a \emph{touching pair} of edges in a self-touching configuration of a graph $G$ to be two edges $e$ and $f$ of $G$ such that these two edges are consecutive in the vertical ordering in the magnified view of at least one edge in $H$. Each touching pair can be assigned to a single face of $G$, the face that lies between the two edges.

We may define a flat folding of a face $F$ of the given graph $G$ to be simply a flat folding of $F$ viewed as a subgraph of $G$. However, when we count flat foldings, we will use a variation of this concept that gives a different number of foldings. In each flat folding of $F$, viewed as a subgraph,
the touching pairs can be partitioned into two subsets: the ones assigned to the interior of $F$ and the ones assigned to the exterior. We consider two flat foldings of (the interior of) face $F$ to be equivalent when they have the same touching pairs assigned to their interior, and distinct otherwise,
ignoring any differences in their exterior touching pairs.

\subsection{Where is the line?}
Note that, although it is intuitively convenient to think of flat foldings as mapping the input graph $G$ to a horizontal line, the line itself does not appear in our definition of a flat folding. It is instead replaced by the path graph $H$. A suitable choice of the graph $H$, together with a mapping from $H$ to a horizontal line may be determined easily for any connected graph $G$, by the following process:
\begin{itemize}
\item Map one of the vertices of $G$ arbitrarily to a point on the horizontal line.
\item Choose an arbitrary direction (left or right) for one of the edges of $G$.
\item Walk from edge to edge along $G$ using the lengths of the edges to determine where along the line to place each vertex of $G$ with respect to the previous vertex in the walk, and using the given turning angles in $G$ to determine for each edge in the walk whether to walk left or right from the previous vertex.
\item Let $H$ be a path graph that connects the resulting points of the line as vertices, in their sorted order along the line.
\end{itemize}
The resulting graph $H$ has infinitely many metric embeddings onto the line (given by the infinitely many choices for the placement of the first vertex of $G$ and the two choices for the direction of the first edge in the walk). However, because flat foldings are defined by self-touching configurations that map $G$ to $H$ rather than by mappings directly to the line, our arbitrary choice of embedding of $H$ onto the line does not affect the existence or number of flat-foldings.

In particular, it is meaningless (with these definitions) to ask whether a flat folding is the same when we reverse it by a reflection of the line: the reflection changes the arbitrary mapping from $H$ to the line, but this mapping is not part of the flat folding. It is also meaningless to ask whether a flat folding is the same when we reverse it by a reflection across the line (reversing its vertical ordering while preserving its horizontal placement); such a reversal is not allowed, as it would reverse the clockwise ordering given in the input embedding of $G$.

\section{Local Characterization}
\label{sec:local-characterization}

In this section we show that for a plane graph with assigned lengths and consistent angles, being able to fold the whole graph flat is equivalent to being able to fold each of its faces flat.

\begin{theorem}
\label{thm:local}
Let $G$ be a plane graph 
with given edge lengths and a consistent angle assignment.
Then $G$ has a flat folding if and only if every face cycle of $G$ (with the induced assignment of lengths and angles) has a flat folding. More strongly, for every combination of flat foldings of the faces of $G$, there exists a flat folding of $G$ itself whose touching pairs for each face are exactly the ones given in the folding of that face.
\end{theorem}

\begin{proof}
One direction is straightforward: if $G$ has a flat folding, then restricting to the faces of $G$ gives flat foldings of the faces with the same touching pairs.

For the other direction, assume we have flat foldings of the faces of $G$. We will show that $G$ has a flat folding with the same touching pairs.
We may assume without loss of generality that $G$ is connected, for otherwise the result follows by finding a flat folding separately for each connected component of~$G$.
As described in \autoref{sec:related},
the assignment of lengths and angles given with $G$ (together with an arbitrary choice of an $x$-coordinate for one vertex and an orientation for one edge) gives us a unique assignment of $x$-coordinates for the vertices of $G$ in any possible flat folding.
We will start by subdividing all the edges of $G$.  
Take the set of $x$-coordinates of vertices of $G$ and add an extra ``half'' $x$-coordinate at the midpoint between any two consecutive coordinate values.  
Subdivide each edge of $G$ by adding vertices at all the $x$-coordinates in this set. 
The same subdivisions can be made in any flat folding of $G$, so there is no change to the existence or number of flat foldings.
The subdivision does change the set of touching pairs, but two edges of the original graph form a touching pair if and only if two of the edges in the paths they are subdivided into form a touching pair, so the correctness of the part of the theorem about touching pairs carries over.

With $G$ subdivided in this way, we carry out the proof by induction on the number of face angles that are assigned to be $360^\circ$ (\emph{mountain folds}). The base case of the induction is the case in which $G$ has only two such angles, on the outer face.
In this case every face cycle consists of two paths of increasing $x$-coordinates and has a unique flat folding, and it is easy to see that $G$ has a flat folding with the same touching pairs.
(Equivalently, the graph in this case is a directed $st$-plane graph so
it is upward planar with each face drawn as two upward paths.)

If $G$ contains a vertex $v$, and an interior face $f$ in which $v$ is a mountain fold, then let $e$ be one of the two edges of $f$ incident to $v$, the one that is uppermost in the magnified view of the flat folding edge corresponding to these two edges, and let $e'$ be the edge immediately above that one. Edge $e'$ must exist, because if $e$ were the topmost edge in this magnified view, then $f$ would necessarily be the exterior face.   
(For example, in \autoref{fig:selftouch}, vertex $g$ is a mountain fold in a cycle; edge $bg$ is the uppermost edge incident to $g$; and $bc$ is the edge immediately above it.)
Let $v'$ be the endpoint of $e'$ whose $x$-coordinate is the same as $v$. We form a graph $G'$ by identifying $v$ with $v'$, ordering the edges of the combined super-vertex so that $e'$ and $e$ remain consecutive. 
This produces a graph, not a multigraph, 
because the other endpoints of $e$ and $e'$ are subdivision points at a ``half'' $x$-coordinate, and so cannot coincide with each other.
(In the example, we would identify vertices $g$ and $c$; the figure does not show the extra subdivision points.)
This 
vertex identification
reduces the number of mountain folds by one compared with $G$, and splits $f$ into two simpler faces $f_1$ and $f_2$. The same split operation can be done to the flat folding of $f$, giving flat foldings of $f_1$ and $f_2$. Thus, $G'$ meets the conditions of the theorem and has fewer mountain folds; by induction it has a flat folding realizing all the touching pairs of its face foldings, which are the same as the touching pairs of the face foldings of $G$. In this flat folding, the supervertex of $G'$ formed from $v$ and $v'$ can be split back into the two separate vertices $v$ and $v'$, giving the desired flat folding of~$G$.

The case when there exist three or more mountain folds on the exterior face is similar, but we must be more careful in our choice of $v$. Each mountain folded vertex on the exterior face is either a local minimum or local maximum of $x$-coordinates; because there are three or more of them, we may choose $v$ to be a vertex that is not a unique global extremum. Then, as above, we find a vertex $v'$ with the same coordinate, above or below $v$, and merge $v$ and $v'$ into a single vertex, giving a graph $G'$ with fewer mountain folds in which the outer face has been split into two faces, one outer and one inner. As before, these two faces inherit a flat folded state from the given flat folding of the outer face of $G$, so by induction $G'$ has  a flat folding. And as before, $v$ and $v'$ may be split back into separate vertices in this flat folding, giving the desired flat folding of~$G$.
\end{proof}

\section{Euler Tours for Nonsimple Faces}
\label{sec:euler}

In the previous section we showed that flat foldability of the whole graph depends only on flat foldability of the faces.  
For the purpose of algorithms, in order to allow the algorithm to name each position of a cycle without the complications of repeated vertices, it is convenient to
 further reduce to faces that are simple cycles.  In particular, 
our algorithms for finding and counting flat foldings of graphs rely on subroutines for solving the same problem on simple cycles.
However, if a plane graph is not $2$-vertex-connected, its faces may not actually be simple cycles: if an articulation vertex or bridge edge appears on the boundary of a face, it may appear multiple times on that face. As an extreme instance of this phenomenon, we may have a tree as our input graph; in this case, it has a single face, with every edge repeated twice and every vertex repeated a number of times equal to its degree.

In this case, we may obtain a collection of simple cycles that represent the faces of the graph as follows. Because $G$ is a plane graph (that is, a topologically but not geometrically embedded planar graph) each side of each edge has a well-defined face incident with it. (This may possibly be the same face as the other side of the same edge, but in such cases we consider the two incidences as distinct from each other.) Each face of $G$ can be described combinatorially as a cyclic (clockwise) ordering of sides of edges. For each face $f$ of $G$, we form a simple cycle of new edges and vertices, having the same cyclic sequence of edge lengths and angles as the corresponding sides of edges and vertices in $f$. The collection of simple cycles formed in this way, for all of the faces, includes two copies of each edge of~$G$, one for each of its two sides.

We call this the \emph{Euler tour technique} by analogy to the Euler tour technique for parallel graph algorithms~\cite{TarVis-FOCS-84}, in which one transforms a tree into a cycle by making two copies of each edge and following an Euler tour of the resulting multigraph. This tour traces around the exterior of a plane embedding of the tree. Here, we similarly make two copies of each edge and multiple copies of each vertex, but we partition those copies into multiple cycles, one for each face of~$G$.

\begin{figure}[t]
\centering\includegraphics[width=4in]{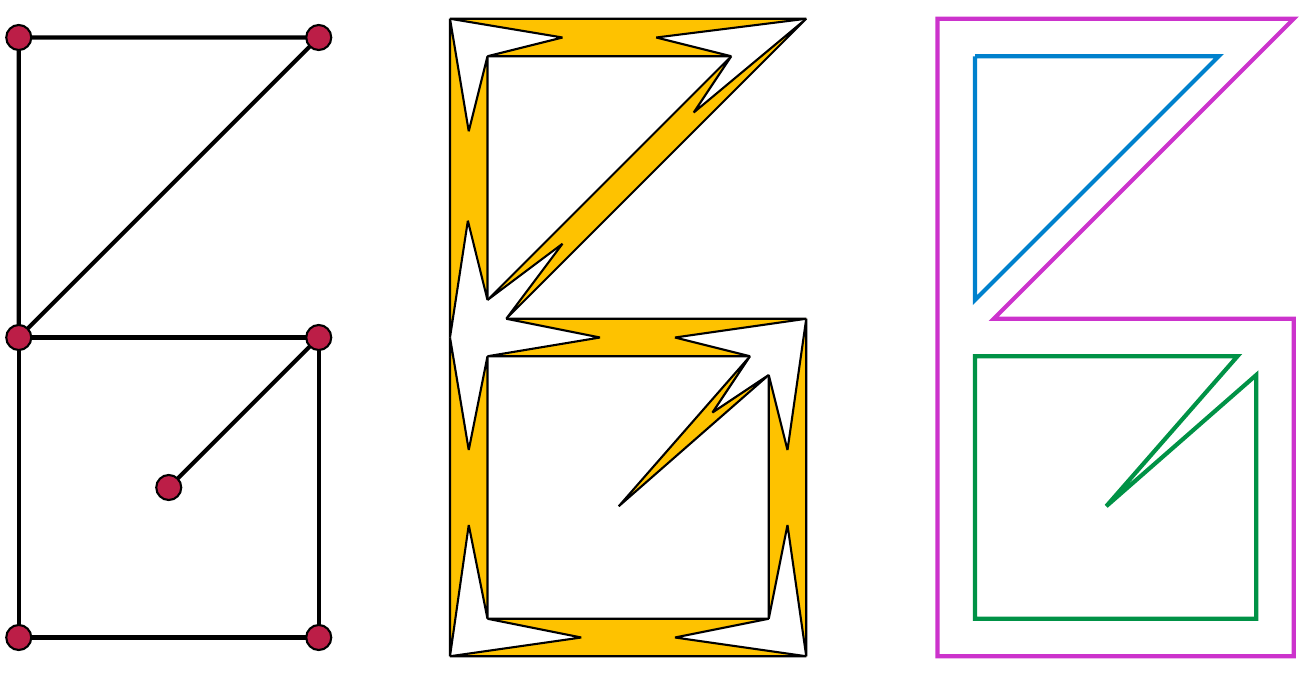}
\caption{Left: a graph $G$ that is not $2$-vertex-connected, and has repeated vertices and edges on its three faces. Center: replacing each dangling edge by a 4-cycle, and each other edge by a 6-cycle, produces a 2-vertex-connected graph $H$ with an equivalent set of flat foldings. Right: The Euler tours of the faces of $G$ are the simple cycles forming the corresponding faces of~$H$.}
\label{fig:eulerize}
\end{figure}

\begin{lemma}
\label{lem:eulerize}
Each face of the given graph $G$ is flat-foldable if and only if each of the simple cycles resulting from the Euler tour technique described above is flat-foldable. The number of flat-foldings of a connected graph $G$ is equal to the product, over these simple cycles, of the number of different patterns of internal visibility that can be formed from a flat-folding of the cycle.
\end{lemma}

\begin{proof}
Given the graph $G$, we form a new graph $H$ from $G$, by replacing each edge of $G$ by a cycle, as shown in \autoref{fig:eulerize}. In more detail:
\begin{itemize}
\item A dangling edge~$e$, i.e., an edge having one endpoint of degree one, becomes a concave four-cycle, with one vertex at the position of that degree-one endpoint adjacent to two vertices at the other endpoint, and with one concave vertex inside the triangle formed by these other three edges. We assign the two edges connecting to the degree-one endpoint the same lengths as the original length of $e$, and the other two edges are assigned a shorter length $\epsilon$. Within this four-cycle, the angles are assigned as a valley fold~($0$) at the three convex vertices and a mountain fold~($360^\circ$) at the concave vertex.
\item An edge $e$ that is not dangling is replaced by a non-convex hexagon; two opposite edges of this hexagon, of length equal to $e$, connect a pair of new vertices at each endpoint of $e$, and each of these pairs is connected by a path of two length-$\epsilon$ edges through a new vertex between the two long edges of the hexagon. Within the face formed by this new hexagon, we assign the angles as a valley fold at the four convex vertices and a mountain fold at the concave vertices.
\item At each vertex $v$ of $G$ of degree $d>1$, identify pairs of vertices from the edge polygons to form a face 
of $2d$ vertices, alternating between vertices placed at the original location of $v$ and vertices that (in the folded state) should be $\epsilon$ away from that location. Label the angles of this face with valley folds at the vertices $\epsilon$ distance away from $v$. Each remaining vertex of this face is formed by identifying a pair of vertices from the polygons of two edges $e$ and $f$. Let $\theta$ be the angle formed at $v$ by edges $e$ and $f$, and within the face for $v$ label the corresponding angle by $360^\circ-\theta$.
\end{itemize}
This system of labels ensures that, at all vertices of the expanded graph, the angle labels correctly add to~$360^\circ$.
Each new $4$-cycle or $6$-cycle in the expanded graph (corresponding to an edge of the original graph) is flat-foldable in a unique way, so the expansion of the edges of $G$ into cycles does not affect the flat-foldability or the number of flat foldings of the resulting graph.
Similarly, each of the $2d$-cycles at a vertex of degree $d$ is uniquely foldable and affects neither flat-foldability nor the number of foldings.
Each face of the resulting graph is either an Euler tour of a face of $G$, one of the new $4$-cycles (from dangling edges of $G$) or $6$-cycles (non-dangling edges), or one of the cycles from vertices of $G$. 
The result follows by applying \autoref{thm:local} to the expanded graph.
\end{proof}

The assumption of connectedness in the lemma is necessary in order for the number of flat foldings to be finite (when nonzero), because for a disconnected graph sliding one of the flat-folded components left or right with respect to the others would result in infinitely many flat-foldings.

\section{Algorithm to Find a Folding}
\label{sec:cycle-folding}

For completeness, we briefly describe a greedy ``crimping'' strategy for finding flat-folded states of simple cycles with pre-assigned edge lengths and fold angles. 
Bern and Hayes~\cite{BerHay-SODA-96} used a similar strategy to flat-fold cycles without pre-assigned angles. Arkin et al.~\cite{ArkBenDem-CGTA-04} applied this method to open polygonal chains with assigned angles. The version here for cycles with assigned angles is described by Demaine and O'Rourke~\cite[Theorem 12.1.6]{DemORo-GFA-07}.
We do not describe its (non-trivial) correctness proof.

First, remove any flat folds from the input by merging the edges on either side of the fold.
Then, repeatedly find an edge $e$ such that the two edges on either side of $e$ are at least as long as $e$, with 
an angle of $0^\circ$ at one end of $e$ and an angle of $360^\circ$ at the other end.
If no such edge $e$ exists, the cycle has no folding. If an edge $e$ that meets these conditions can be found, it is safe to perform both folds, merging $e$ with its two neighboring edges into a single edge of a simpler polygon.

Maintaining a set of edges that are ready to be folded, and performing each fold, takes constant time per fold, so folding a cycle in this way may be done in linear time.
Putting the characterization from \autoref{sec:local-characterization} together with the algorithm for folding a single cycle described above gives us an algorithm for testing whether a given plane graph $G$ with edge length and angle assignment is flat foldable:

\begin{theorem}
\label{thm:decision}
We can test flat foldability of a plane graph with given edge lengths and given angle assignment in linear time.
\end{theorem}

\begin{proof}
We partition the graph into its component faces, and apply the crimping algorithm to an Euler tour of each face. Each face takes time proportional to its size, so the total time 
is linear. The correctness of forming simple cycles from each face by taking Euler tours follows from \autoref{lem:eulerize}.
\end{proof}

\section{Counting Flat Foldings}

We cannot use crimping to count the flat foldings of a cycle, because some flat foldings cannot be formed by a sequence of crimping steps; \autoref{fig:uncrimpable} depicts an example. Additionally, the crimping procedure that we use finds only a single folding, and it would not be clear how to modify it even to count all the foldings that arise from repeated crimping. 
Instead, to count flat foldings in a single cycle, we use dynamic programming. 

\begin{figure}[htb]
\centering\includegraphics[width=4in]{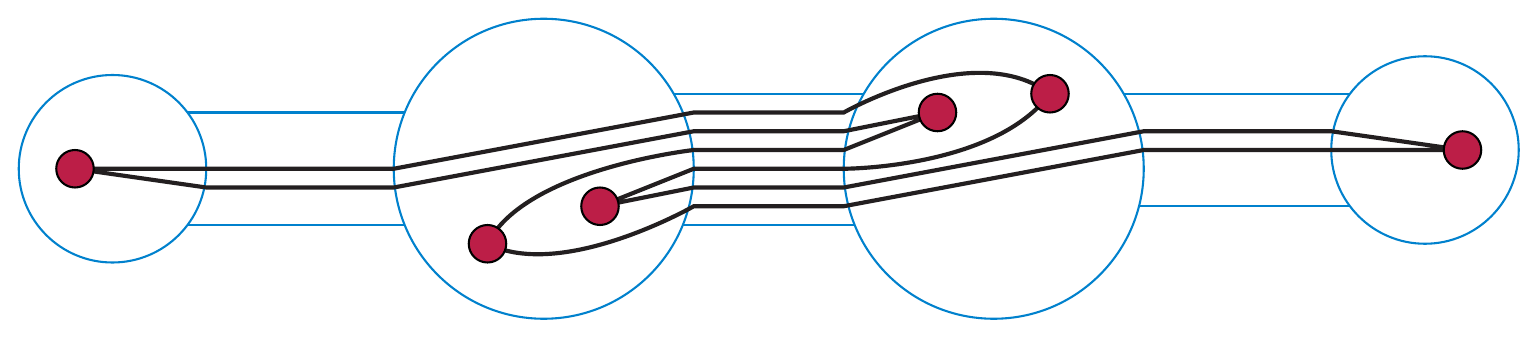}
\caption{Magnified view of a flat folding that cannot be obtained by crimping.}
\label{fig:uncrimpable}
\end{figure}

As a preliminary result for our counting algorithm, we prove that the number of flat foldings of the interior of a simple cycle (with assigned lengths and angles) may be counted by dynamic programming in $\tilde O(n^4)$ time. Here the $\tilde O$ notation means that we omit logarithmic factors in the time bound, which in our case come from the nonlinear complexity of multiprecision integer multiplication. More precisely, let $M(n)=\tilde O(n)$ denote the time to multiply two $O(n)$-bit binary numbers.  Then our time bound is $O(n^3 M(n))=\tilde O(n^4)$.

\begin{lemma}
\label{lem:num-bits}
Given a single $n$-vertex cycle, with an assignment of edge lengths and angles, the number of flat foldings of the interior of the cycle can be counted by a number whose binary representation has $O(n)$ bits.
\end{lemma}

\begin{proof}
A flat folding is determined by the pairs of edges that have a nonzero length directly visible to each other. To represent this information, create an auxiliary graph that has $n$ vertices, one at each midpoint of one of the $n$ cycle edges. Connect two of these vertices by an edge in the auxiliary graph when they correspond to two cycle edges that are visible to each other. This auxiliary graph can be drawn without crossings within the magnified view of the flat folding, by routing each edge of the auxiliary graph along each of the two cycle edges that it connects, to the point where the two edges are directly above each other. Therefore, it forms an outerplanar graph on its $n$ vertices.

We can bound the number of outerplanar graphs on these vertices, with their fixed cyclic ordering,
by multiplying the number of maximal outerplanar graphs on the same vertex set (a Catalan number) by the number of different subsets of edges that each of these maximal outerplanar graphs has ($2^{2n-3}$). The result follows from the fact that this bound is single-exponential in~$n$. For related results on the enumeration of outerplanar graphs, see Bodirsky et al.~\cite{BodFus-EJC-07} and their references.
\end{proof}

\begin{lemma}
\label{lem:count-cycle}
Given a single $n$-vertex cycle, with an assignment of edge lengths and angles, it is possible to count the flat foldings of the interior of the cycle in time $O(n^3 M(n))=\tilde O(n^4)$.
\end{lemma}

\begin{proof}
We use the given assignment to compute $x$-coordinates for all vertices, and subdivide the cycle as in the proof of \autoref{thm:local}. After this step, no mountain folded vertex is interior to the range of $x$-coordinates of another edge, and  each input edge is subdivided into a path of at least two edges.
By the argument used to prove \autoref{thm:local}, once this subdivision is fixed, a flat-folding of the cycle is equivalent
to a partition of the edges of the subdivided cycle into touching pairs (with both edges in each touching pair having the same endpoint coordinates) such that any two touching pairs are disjoint or nested.

To count sets of touching pairs of this type, we use a dynamic programming algorithm that counts, for each ordered pair $u,v$ of vertices with the same $x$-coordinate, 
the numbers $A(u,v)$ and $B(u,v)$ of flat-foldings (if any) of the oriented polygon formed by identifying $u$ with $v$ and keeping the part of the input from $u$ clockwise to $v$. Here we only consider pairs such that the edge $e$ immediately clockwise from $u$ and the edge $f$ immediately counterclockwise from $v$ both extend in the same direction, left or right, and such that the identification of $u$ with $v$ produces a valley fold rather than a mountain fold. $A(u,v)$ will count the total number of flat-foldings of this polygon, while $B(u,v)$ will count only a subset of the flat-foldings, the ones in which edges $e$ and $f$ are visible to each other rather than being blocked by a mountain-fold vertex with the same $x$-coordinate as both $u$ and $v$. By symmetry we can assume without loss of generality that edges $e$ and $f$ both lie to the left of $u$ and $v$, so in any flat-folded state $u$ and $e$ must be below $v$ and $f$.
See \autoref{fig:dyn-prog-1}.

\begin{figure}[htb]
\vspace{-2ex}
\centering
\includegraphics[width=5.25in]{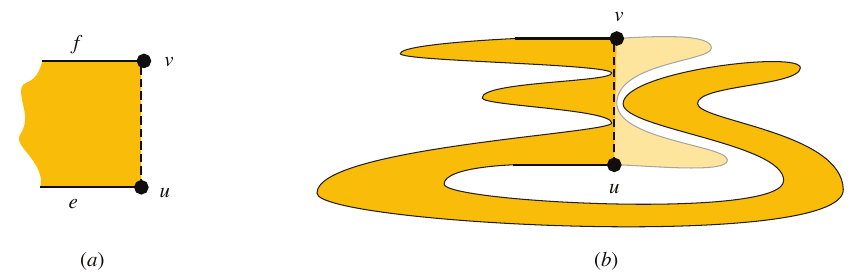}
\caption{(\emph{a}) The dynamic programming subproblem for $u$ and $v$.  (\emph{b}) A possible solution to the subproblem, and its place in the whole solution (light colour).}
\label{fig:dyn-prog-1}
\end{figure}

\begin{figure}[b]
\vspace{-1ex}
\centering
\includegraphics[width=4.25in]{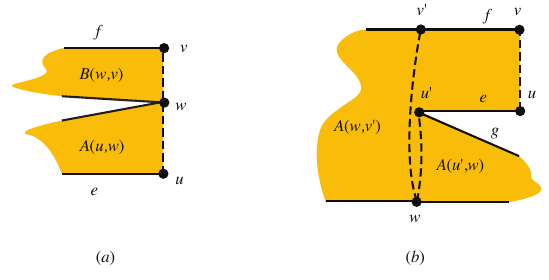}
\caption{(\emph{a}) Computing $A(u,v)$ by trying all possible choices for $w$.  
(\emph{b}) Computing $B(u,v)$ in case $u'$ is a mountain but $v'$ is not.}
\label{fig:dyn-prog-2}
\end{figure}

To compute $A(u,v)$, we sum $B(u,v)$ together with the number of folds in which a mountain vertex $w$ blocks $e$ from $f$. We will check all possible choices of $w$; to avoid double-counting, we will only consider folds where $w$ is the topmost such mountain vertex (the one closest to $v$). Thus, for each choice of $w$, we include in our sum the term $A(u,w)B(w,v)$ counting the number of folds in the two sub-polygons between $u$ and $w$ and between $w$ and $v$. That is, we may calculate $A(u,v)$ by the formula
\[
A(u,v)=B(u,v)+\sum_w A(u,w)B(w,v)
\]
where the sum is over the possible choices of mountain vertices that are between $u$ and $v$ in the cyclic ordering of the cycle, have their two adjacent edges to the left, and have the same $x$-coordinate as $u$ and $v$.
See \autoref{fig:dyn-prog-2}$(a)$.

To compute $B(u,v)$, we consider the two vertices $u'$ and $v'$ at the other ends of the incident edges $e$ and $f$. Because of the way we subdivided the input, $u'$ and $v'$ again have the same $x$-coordinate. 
If $u'=v'$, we are done and there is exactly one flat-folded state: $B(u,v)=1$. If either $u'$ or $v'$ is a valley fold, there can be no valid flat-folded state of the polygon between $u$ and $v$, and again we are done: $B(u,v)=0$. And if $u'$ and $v'$ are both flat vertices, then we can copy the value from a simpler subproblem: $B(u,v)=A(u',v')$.

In the remaining cases, one or both of $u'$ and $v'$ is a mountain fold. Consider the case in which $u'$ is a mountain fold and $v'$ is not; the other cases are similar. 
See \autoref{fig:dyn-prog-2}$(b)$.
Let $g$ be the other edge incident to $u'$. Because $u'$ is a mountain fold, edge $g$ lies to the right of $u'$.  In any flat-folded state, $u'$ must see another non-mountain vertex $w$ below it, from the sub-polygon between $u$ and $v$ (ignoring intervening mountains). Identifying $u'$ and $w$ splits the polygon from $u$ to $v$ into two parts, one to the right of the identified vertex and one to the left. The right sub-polygon can be described by a subproblem from $u'$ to $w$, and the left sub-polygon can be described by a subproblem from $w$ to $v'$. Thus, in this case, we have the formula
\[
B(u,v)=\sum_w A(u',w)A(w,v')
\]
where again the sum is over all suitable choices of $w$.

In the case where both $u'$ and $v'$ are mountain folds, we have a similar sum over two choices of intermediate vertex, one below $u'$ and one above $v'$, forming two right-pointing subproblems with one left-pointing subproblem between them:
\[
B(u,v)=\sum_w \sum_x A(u',w) A(w,x) A(x,v').
\]
However, we may avoid summing over all pairs of choices of $w$ and $x$ by factoring this expression. To do so, define
\[
C(u',x)=\sum_w A(u',w)A(w,x).
\]
Then we may rewrite the expression for $B(u,v)$ (in the case of two mountain folds) as
\[
B(u,v)=\sum_x C(u',x)A(x,v').
\]

The number of flat-foldings of a whole polygon, when it is an interior face, can be found as $A(v,v)$ where $v$ is a rightmost vertex of the polygon (one that maximizes its $x$-coordinate). For the exterior face, we may reduce to the case of an interior face by letting $v$ be a rightmost vertex of the polygon, placing a new artificial vertex $u$ to the left of all vertices of the polygon (that is, with smaller $x$-coordinate than any vertex), breaking the polygon at $v$ (splitting it into two copies $v$ and $v'$), and rejoining it by two edges $v$--$u$--$v'$ that, in any flat-folding, necessarily surround the rest of the polygon. For this augmented polygon, we can compute the number of flat-foldings as $A(v,v')$.

Finally, we analyze the running time.  The number of subproblems, i.e., the number of pairs $(u,v)$, is $O(n^3)$ because there are $O(n)$ coordinates each with $O(n)$ vertices (after subdividing edges).
The number of multiplications performed in the calculation of $A(u,v)$ can be bounded as $O(n^2)$ choices of the pair $(u,v)$ for each mountain fold $w$, for a total of $O(n^3)$ multiplications.
each subproblem involves a summation with a linear number of choices. The calculation of $B(u,v)$ involves a constant or linear number of choices depending on whether one or both of $u'$ or $v'$ is a mountain fold; the number of cases with a linear number of choices, in which at least one is a mountain fold, is $O(n^2)$. The arithmetic operations in these summations involve numbers with a linear number of bits (\autoref{lem:num-bits}). Therefore, the total time is $O(n^3 M(n))=\tilde O(n^4)$.
\end{proof}

\begin{theorem}
We can count the flat foldings of a connected $n$-vertex planar graph $G$ with an assignment of edge lengths and angles in time $O(n^3 M(n))=\tilde O(n^4)$.
\end{theorem}

\begin{proof}
We apply \autoref{lem:count-cycle} to the Euler tour of each face of $G$ and return the product of the resulting numbers. Because the number of bits needed to represent the number of foldings of each face is proportional to the complexity of the face, the number of bits needed to represent the number of foldings of the whole graph is proportional to the sum of the complexities of the faces, which is $O(n)$. The time to calculate the product is dominated by the time to perform the dynamic programming algorithm within each face, which is maximized when a constant number of faces have linear complexity, giving the time bound claimed.
\end{proof}

\textbf{Acknowledgements.}
This research was performed in part at the 29th Bellairs Winter Workshop on Computational Geometry. A preliminary version of these results were presented at the 22nd International Symposium on Graph Drawing~\cite{AbeDemDem-GD-14}.
Erik Demaine thanks Ilya Baran and Muriel Dulieu, and the authors of
\cite{AbeDemDem-IJCGA-13}, for many discussions attempting to solve
this problem. We also thank Jason Ku for helpful comments on a draft of this paper.
Erik Demaine was supported in part by NSF ODISSEI grant EFRI-1240383 and NSF Expedition grant CCF-1138967.
David Eppstein was supported in part by NSF grants  CCF-1228639, CCF-1618301,
and CCF-1616248 and by ONR grant N00014-08-1-1015.
Anna Lubiw was supported in part by the Natural Sciences and Engineering Research Council of Canada (NSERC).

\raggedright
\bibliographystyle{abuser}
\bibliography{folding}

\end{document}